\documentclass[amsmath,amssymb,aps,pra,twocolumn]{revtex4}

\usepackage{amsthm}
\usepackage{url}

\newcommand{\comment}[1]{}
\newcommand{\ket}[1]{|#1\rangle}
\newcommand{\bra}[1]{\langle #1|}
\newcommand{\braket}[2]{\langle #1 | #2 \rangle}

\newcommand{\eps}{\epsilon}

\newcommand{\tr}{\operatorname{tr}}

\newcommand{\ot}{\otimes}
\newcommand{\cC}{\mathcal{C}}

\newcommand{\cP}{\mathcal{P}}
\newcommand{\cS}{\mathcal{S}}
\newcommand{\cU}{\mathcal{U}}

\renewcommand\Re{\operatorname{Re}}

\newcommand{\be}{\begin{equation}}
\newcommand{\ee}{\end{equation}}
\newcommand{\bes}{\begin{equation*}}
\newcommand{\ees}{\end{equation*}}
\def\ba#1\ea{\begin{align}#1\end{align}}
\def\bas#1\eas{\begin{align*}#1\end{align*}}
\def\bit{\begin{itemize}}
\def\eit{\end{itemize}}
\def\l{\left}
\def\r{\right}
\def\<{\langle}
\def\>{\rangle}

\newtheorem{theorem}{Theorem}
\newtheorem{lemma}[theorem]{Lemma}
\newtheorem{definition}[theorem]{Definition}

\newcommand{\thmref}[1]{Theorem \ref{thm:#1}}

\newcommand{\lemref}[1]{Lemma \ref{lem:#1}}
\newcommand{\secref}[1]{Section \ref{sec:#1}}
\newcommand{\defref}[1]{Definition \ref{def:#1}}

\begin{document}

\title{Learning and Testing Algorithms for the Clifford Group}

\author{Richard A.~Low}
 \email{low@cs.bris.ac.uk}
\affiliation{Department of Computer Science, University of Bristol, Bristol, U.K.}


\date{\today}

\begin{abstract}
Given oracle access to an unknown unitary $C$ from the Clifford group and its conjugate, we give an exact algorithm for identifying $C$ with $O(n)$ queries, which we prove is optimal.  We then extend this to all levels of the Gottesman-Chuang hierarchy (also known as the $\cC_k$ hierarchy).  Further, for unitaries not in the hierarchy itself but known to be close to an element of the hierarchy, we give a method of finding this close element.  We also present a Clifford testing algorithm that decides whether a given black-box unitary is close to a Clifford or far from every Clifford.
\end{abstract}

\maketitle

\section{Introduction}

A central problem in quantum computing is to determine an unknown quantum state from measurements of multiple copies of the state.  This process is known as quantum state tomography (see \cite{NielsenChuang} and references therein).  By making enough measurements, the probability distributions of the outcomes can be estimated from which the state can be inferred.  A related problem is that of quantum process tomography, where an unknown quantum evolution is determined by applying it to certain known input states.  There are several methods for doing this, including what are known as Standard Quantum Process Tomography \cite{ChuangNielsen97, PoyatosCiracZoller97} and Ancilla Assisted Process Tomography \cite{DArianoPresti01, Leung03}.  These methods work by using state tomography on the output states for certain input states.

However, all these procedures share one important downside: the number of measurements required increases exponentially with the number of qubits.  This already presents problems even with systems achievable with today's technology, for which complete tomographical measurements can take hours (e.g.~\cite{HaffnerTomography}) making tomography of larger systems unfeasible.  Unfortunately this exponential cost is necessary to determine a completely unknown state or process, since there are exponentially many parameters to measure.  To make tomography feasible for larger systems, we need to find a restriction that requires fewer measurements, ideally polynomially many.

One way to improve the measurement, or query, complexity is to assume some prior knowledge of the process.  For example, suppose the process was known to be one of a small number of unitaries, then the task is just to decide which.  This is the approach we take here.  As a simple example, consider being given a black box implementing an unknown Pauli matrix.  By applying this to half a maximally entangled state, the Pauli can be identified with one query.  This is essentially superdense coding \cite{SuperdenseCoding} and is explained in \secref{LearningPaulis}.  Indeed, if the black box performed a tensor product of arbitrary Paulis on $n$ qubits then it too can be identified with just one query.

We extend this to work for elements of the Clifford group (the normaliser of the Pauli group; see \defref{CliffordGroup}) and show that any member of the Clifford group can be learnt with $O(n)$ queries, which we show is optimal.  The Clifford group is an important subgroup of the unitary group that has found uses in quantum error correction and fault tolerance \cite{QECGeometry,ShorFaultTolerance,GottesmanFaultTolerance}.

Then generalising further, we show that elements of the Gottesman-Chuang hierarchy \cite{GottesmanChuang} (see \defref{GottesmanChuangHierarchy}), also known as the $\cC_k$ hierarchy, can also be learnt efficiently.  As the level $k$ increases, the set $\cC_k$ includes more and more unitaries so this implies ever larger sets can be learnt, although the number of queries scales exponentially with $k$.  Our methods also work if the unitary is known to be close to a Clifford (or any element of $\cC_k$ for some known $k$) rather than exactly a Clifford.

We also give a Clifford testing algorithm, which determines whether an unknown unitary is close to a Clifford or far from every Clifford.  This is an extension of the Pauli testing algorithm given in \cite{QBF}.  Indeed, our results are closely related to results in \cite{QBF} and we use some of the algorithms presented there as ingredients.  Our results can also be compared with \cite{AaronsonLearnability}, which contains methods to approximately learn quantum states.  Another related result is that of Aaronson and Gottesman \cite{AaronsonGottesmanStabilisers}, which provides a method of learning stabiliser states with linearly many copies.

We only consider query complexity although, at least for the Clifford group results, our methods are computationally efficient too.

The rest of the paper is organised as follows.  In \secref{definitions}, we define the Pauli and Clifford groups and the Gottesman-Chuang hierarchy.  In \secref{ExactLearning} we present our algorithm for exact learning of Clifford and $\cC_k$ elements.  In \secref{ApproxLearning} we show how to find the closest element of $\cC_k$ to an unknown unitary.  In \secref{CliffordTesting} we present our Clifford testing algorithm and then conclude in \secref{Conclusions}.

\section{The Pauli and Clifford Groups and the Gottesman-Chuang Hierarchy}
\label{sec:definitions}

Firstly, we define the Pauli group.  The Pauli operators are
\ba
\sigma_I &= \begin{pmatrix}1 & 0 \\ 0 & 1\end{pmatrix} 
\qquad
\sigma_x = \begin{pmatrix} 0& 1 \\ 1 & 0\end{pmatrix} \nonumber \\
\sigma_y &= \begin{pmatrix} 0& -i \\ i & 0\end{pmatrix} 
\qquad
\sigma_z = \begin{pmatrix} 1& 0 \\ 0 & -1\end{pmatrix} 
\ea
which we extend to $n$ qubits by taking all tensor products of these one qubit matrices.  Call this set $\hat{\cP}$ with elements $\sigma_p$ with $p \in \{I, x, y, z\}^n$.  We then have $|\hat{\cP}| = 4^n$ and $\tr \sigma_p \sigma_q = 2^n \delta_{pq}$.  Also, the Pauli matrices form an orthogonal basis for $2^n \times 2^n$ matrices.  Therefore we can write any such matrix in the Pauli basis in the form $\sum_p \gamma(p) \sigma_p$.  Then to make this into a group, the Pauli group $\cP$, we must include each matrix in $\hat{\cP}$ with phases $\{\pm 1, \pm i\}$.

We can now define the Clifford group:
\begin{definition}[The Clifford group]
\label{def:CliffordGroup}
The Clifford group is the normaliser of the Pauli group i.e.
\bes
\cC = \{ U \in \cU(2^n) : U \cP U^\dagger \subseteq \cP \}.
\ees
\end{definition}

Then the Gottesman-Chuang hierarchy is a generalisation:
\begin{definition}[The Gottesman-Chuang hierarchy \cite{GottesmanChuang}]
\label{def:GottesmanChuangHierarchy}
Let $C_1$ be the Pauli group $\cP$.  Then level $C_k$ of the hierarchy is defined recursively:
\bes
\cC_k = \{ U \in \cU(2^n) : U \cP U^\dagger \subseteq \cC_{k-1} \}.
\ees
\end{definition}
By definition, $\cC_2$ is the Clifford group $\cC$.  For $k>2$, $\cC_k$ is no longer a group but contains a universal gate set, whereas $\cC_1$ and $\cC_2$ are not universal.

\section{Learning Gottesman-Chuang Operations}
\label{sec:ExactLearning}

Before we give our algorithm for learning Gottesman-Chuang operations, we present a simple method for learning Pauli operations, which we use as the main ingredient.

\subsection{Learning Pauli Operations}
\label{sec:LearningPaulis}

This is due to \cite{QBF} and is in fact identical to the superdense-coding protocol \cite{SuperdenseCoding}.
\begin{theorem}[\cite{QBF}, Proposition 20]
\label{thm:PauliLearning}
Pauli operations can be identified with one query and in time $O(n)$.
\end{theorem}
\begin{proof}
Apply the operator $\sigma_p$ to half of the maximally entangled state $\ket{\psi} = 2^{-n/2} \sum_i \ket{i i}$.  For different choices of $\sigma_p$, the resulting states are orthogonal so can be perfectly distinguished:
\bas
\bra{\psi} \l(\sigma_p \ot I\r) \l(\sigma_q \ot I\r) \ket{\psi} &= 2^{-n} \sum_{ij} \bra{ii} \sigma_p \sigma_q \ot I \ket{jj} \\
&= 2^{-n} \sum_{ij} \bra{i} \sigma_p \sigma_q \ket{j} \braket{i}{j} \\
&= 2^{-n} \sum_{i} \bra{i} \sigma_p \sigma_q \ket{i} \\
&= 2^{-n} \tr \sigma_p \sigma_q \\
&= \delta_{pq}.
\eas
The time complexity $O(n)$ comes from the preparation and measurement operations.
\end{proof}

\subsection{Learning Clifford Operations}

We can now present our algorithm for learning Clifford operations to illustrate our main idea for learning unitaries in the Gottesman-Chuang hierarchy.  We will use the fact that knowing how a unitary acts by conjugation on all elements of $\hat{\cP}$ identifies it uniquely (up to phase):
\begin{lemma}
\label{lem:UConjPauli}
Knowing $U \sigma_p U^\dagger$ for all $\sigma_p \in \hat{\cP}$ uniquely determines $U$, up to global phase.
\end{lemma}
\begin{proof}
The Pauli matrices form a basis for all $2^n \times 2^n$ matrices so knowing the action of $U$ on the Paulis is enough to determine the action of $U$ on any matrix up to phase.  The phase cannot be determined because action by conjugation does not reveal the phase.
\end{proof}
Now let $G = \{ \sigma_{x_i}, \sigma_{z_i} \}_{i=1}^n$ where $\sigma_{x_i}$ ($\sigma_{z_i}$) is the matrix with $\sigma_x$ ($\sigma_z$) acting on qubit $i$ and trivially elsewhere.  We think of this as a set of generators for $\hat{\cP}$ since each element of $\hat{\cP}$ can be written as a product of elements of $G$, up to phase.  Using this, knowledge of how $U$ acts on elements of $G$ is sufficient to determine the action on all of $\hat{\cP}$:
\begin{lemma}
\label{lem:UConjGenerator}
$U \sigma_p U^\dagger$ for any $\sigma_p \in \hat{\cP}$ can be calculated from knowledge of $U \sigma_g U^\dagger$ for each $\sigma_g \in G$.
\end{lemma}
\begin{proof}
Let $\sigma_p = \alpha \sigma_{g_1} \ldots \sigma_{g_m}$ for $\sigma_{g_i} \in G$ where $\alpha$ is a phase.  Then
\bes
U \sigma_p U^\dagger = \alpha U \sigma_{g_1} \ldots \sigma_{g_m} U^\dagger = \alpha U \sigma_{g_1} U^\dagger \ldots U \sigma_{g_m} U^\dagger.\qedhere
\ees
\end{proof}
With these definitions and observations, we can now present the Clifford learning algorithm.
\begin{theorem}
\label{thm:CliffordLearning}
Given oracle access to an unknown Clifford operation $C$ and its conjugate $C^\dagger$, $C$ can be determined exactly (up to global phase) with $2n+1$ queries to $C$ and $2n$ to $C^\dagger$.  The algorithm runs in time $O(n^2)$.
\end{theorem}
\begin{proof}
From the definition of the Clifford group, $C \sigma_p C^\dagger \in \cP$ for all $\sigma_p \in \hat{\cP}$.  Note that $C \sigma_p C^\dagger$ is not necessarily a Pauli operator in $\hat{\cP}$ because there is a phase of $\pm 1$ (complex phases are not allowed because $C \sigma_p C^\dagger$ is Hermitian).  Determining which Pauli operator and phase for every $\sigma_p$ would be sufficient to learn $C$ using \lemref{UConjPauli}.  But from \lemref{UConjGenerator}, we only need to know $C \sigma_g C^\dagger$ for each $\sigma_g \in G$.

Let $C \sigma_{x_i} C^\dagger = \alpha_i \sigma_{a_i}$ and $C \sigma_{z_i} C^\dagger = \beta_i \sigma_{b_i}$, where $\alpha_i, \beta_i = \pm 1$.  Knowing just $\sigma_{a_i}$ and $\sigma_{b_i}$ is enough to specify $C$ up to a Pauli correction factor $\sigma_q$ which gives the phases $\alpha_i$ and $\beta_i$.  Choosing $\sigma_q$ that anticommutes with $\sigma_{x_i}$ flips the sign of $\alpha_i$ and similarly for $\sigma_{z_i}$.  We now present the algorithm:
\begin{enumerate}
\item{Apply $C \sigma_{x_i} C^\dagger$ and $C \sigma_{z_i} C^\dagger$ for each $i$ and use \thmref{PauliLearning} to determine $\sigma_{a_i}$ and $\sigma_{b_i}$.  This uses $2n$ queries to both $C$ and $C^\dagger$.}
\item{Let $C'$ be such that $C' \sigma_{x_i} C'^\dagger = \sigma_{a_i}$ and $C' \sigma_{z_i} C'^\dagger = \sigma_{b_i}$ i.e.~the phases are all $+1$.  Then, choosing a phase for $C'$, we can write $C = C' \sigma$ where
\be
\sigma = \prod_{i : \alpha_i = -1} \sigma_{z_i} \prod_{i : \beta_i = -1} \sigma_{x_i}.
\ee
Then implement $C'^\dagger C$ to determine $\sigma$ using \thmref{PauliLearning}.  This uses one query to $C$.  We can now calculate the phases $\alpha_i$ and $\beta_i$.}
\end{enumerate}
To work out the time complexity, note that in step 1 the $O(n)$ time Pauli learning algorithm is called $2n$ times.  Then for step 2, the Clifford $C'$ can be implemented in $O(n^2)$ time using for example Theorem 10.6 of \cite{NielsenChuang}.
\end{proof}
We now show that this algorithm is optimal, in terms of number of queries, up to constant factors:
\begin{lemma}
\label{lem:CliffordLearningOptimal}
Any method of learning a Clifford gate requires at least $n$ queries.
\end{lemma}
\begin{proof}
Each application of the gate $C$ can give at most $2n$ bits of mutual information about $C$.  This follows from the optimality of superdense coding \cite{SuperdenseCoding}.  The Clifford group (modulo global phase) is of size \cite{Calderbank98quantumerror} $2^{n^2+2n+3} \prod_{j=1}^n (4^j-1) \ge 2^{2n^2+n+3}$.  To identify an element with $m$ queries, we therefore need
\be
2^{2nm} \ge 2^{2n^2+n+3}
\ee
which implies $m \ge n$.
\end{proof}
It is unfortunate that access to $C^\dagger$ is also required, but we do not know a method with optimal query complexity that works without $C^\dagger$.  There are however methods that use $O(n^2)$ queries that do not use $C^\dagger$.  The result of \cite{HowManyCopiesStateDesc} can be used to show that $O(n^2)$ queries to $C$ are sufficient, by distinguishing the states $C \ot I \ket{\psi}$ for different Cliffords $C$ and where $\ket{\psi}$ is the maximally entangled state.  We can use \lemref{UniqueCloseC} to show that these states are far apart in the distance measure used in \cite{HowManyCopiesStateDesc}, allowing us to apply their result.

\subsection{Learning Gottesman-Chuang Operations}

\thmref{CliffordLearning} can easily be generalised to learning any operation from the $\cC_k$ hierarchy:
\begin{theorem}
\label{thm:CkLearning}
Given oracle access to an unknown operation $C \in \cC_k$ and its conjugate $C^\dagger$, $C$ can be determined exactly (up to phase) with $\frac{(2n)^k - 1}{2n-1}$ queries to $C$ and $(2n)^{k-1}$ to $C^\dagger$.
\end{theorem}
\begin{proof}
The proof is by induction.  The base case is for the Paulis and is proven in \thmref{PauliLearning}.  Then, to learn $C \in \cC_{k+1}$, we assume we have a learning algorithm for members of $\cC_{k}$.  Apply $C \sigma_g C^\dagger$ for each $\sigma_g \in G$.  These operations are elements of $\cC_{k}$ so use the learning algorithm for $\cC_{k}$ to determine these up to phase.  Then use the last step of \thmref{CliffordLearning} to determine the phases.

We now determine the number of queries to $C$ and $C^\dagger$.  Let $T(k)$ be the number of queries to $C$ and $T'(k)$ the number of queries to $C^\dagger$.  We have the recurrences
\ba
T(k+1) &= 2n T(k) + 1 \nonumber \\
T(1) &= 1
\ea
and
\ba
T'(k+1) &= 2n T'(k) \nonumber \\
T'(2) &= 2n
\ea
which have solutions $T(k) = \frac{(2n)^k - 1}{2n-1}$ and $T'(k) = (2n)^{k-1}$ (with $T'(1) = 0$).
\end{proof}

\section{Learning Unitaries Close to $\cC_k$ Elements}
\label{sec:ApproxLearning}

Here we suppose that we are given a unitary that is known to be close to an element of $\cC_k$ for some given $k$.  We present a method for finding this element.  But first we must define our distance measure.

We would like our distance measure to not distinguish between unitaries that differ by just an unobservable global phase.  We define a `distance' $D$ below with this property.  However, firstly define the distance $D^+$ to be a normalised 2-norm distance:
\begin{definition}
For $U_1$ and $U_2$ $d \times d$ matrices,
\bes
D^+(U_1, U_2) := \frac{1}{\sqrt{2d}} || U_1 - U_2 ||_2.
\ees
where $||A||_2 = \sqrt{\tr A^\dagger A}$.
\end{definition}
We have chosen the normalisation so that $0 \le D^+(U_1, U_2) \le 1$.  We now define our phase invariant `distance':
\begin{definition}
\label{def:DistanceMeasure}
For $U_1$ and $U_2$ $d \times d$ matrices,
\bes
D(U_1, U_2) := \frac{1}{\sqrt{2d^2}} || U_1 \ot U_1^* - U_2 \ot U_2^* ||_2
\ees
\end{definition}
This is not a true distance since $D(U_1, U_2) = 0$ does not imply $U_1 = U_2$, but that $U_1$ and $U_2$ are the same up to a phase so the difference is unobservable.  From the 2-norm definition, we can show:
\begin{lemma}
\label{lem:DistFromInnerProd}
\be
D^+(U_1, U_2) = \sqrt{1 - \Re \frac{\tr U_1 U_2^\dagger}{d}}
\ee
and
\be
D(U_1, U_2) = \sqrt{1 - \l| \frac{\tr U_1 U_2^\dagger}{d} \r|^2}.
\ee
\end{lemma}
From this we can easily see that $0 \le D(U_1, U_2) \le 1$ with equality if and only if $U_1$ and $U_2$ are orthogonal.  Further note that by the unitary invariance of the 2-norm, both $D$ and $D^+$ are unitarily invariant and from the triangle inequality for the 2-norm they both obey the triangle inequality.

Our approximate learning method will find the unique closest element of $\cC_k$ to $U$.  In order to guarantee uniqueness, the distance must be upper bounded:
\begin{lemma}
\label{lem:UniqueCloseC}
If $D(U, C) < \frac{1}{2^{k-1/2}}$ for some $C \in \cC_k$ then $C$ is unique up to phase.
\end{lemma}
The proof is in Appendix \ref{sec:ProofLemUniqueCloseC}.
\begin{theorem}
\label{thm:ApproxLearning}
Given oracle access to $U$ and $U^\dagger$ and $k$ such that $D(U, C) \le \eps$ for some $C \in \cC_k$ with
\be
\eps' := \sqrt{2(1-(2^{k-1} \eps)^2)} - 1 > 0
\ee
then $C$ can be determined with probability at least $1-\delta$ with $O\l(\frac{1}{\eps'^2} (2n)^{k-1} \log \frac{(2n+1)^{k-1}}{\delta}\r)$ queries.
\end{theorem}
\begin{proof}
By \lemref{UniqueCloseC}, $C$ is unique up to phase.  We now prove the Theorem by induction.

For $k=1$, use Proposition 21 of \cite{QBF} to learn the closest Pauli operator.  This works by repeating the Pauli learning method \thmref{PauliLearning} and taking the majority vote.  This uses $O\l(\frac{1}{\eps'^2} \log \frac{1}{\delta}\r)$ queries to succeed with probability at least $1-\delta$.

Now for the inductive step.  Assume we have a learning algorithm for level $k$.  Then for $C \in \cC_{k+1}$, let $C \sigma_{g_i} C^\dagger = C_{g_i}$ for $\sigma_{g_i} \in G$.  By \lemref{CloseHaveClosePauli}, we have $D(U \sigma_{g_i} U^\dagger, C_{g_i}) \le 2\eps$.  Use the learning algorithm for level $k$ to determine $C_{g_i}$ up to phase for all $i$.  Then to find the phases we use the same method as before: implement any $C'$ with $C' \sigma_{g_i} C'^\dagger = \pm C_{g_i}$ for any (known) choice of phase.  Then $C'= C \sigma_q$ for some Pauli operator $\sigma_q$.  We can determine $\sigma_q$ by implementing $C'^\dagger U$ and using the $k=1$ learning algorithm since
\ba
D(C'^\dagger U, \sigma_q) &= D(U, C' \sigma_q) \nonumber\\
&= D(U, C) \le \eps.
\ea
Now we calculate the success probabilities and number of queries.  There are $2n+1$ calls to the algorithm at lower levels, which all succeed with probability at least $1-\delta$.  So at this level the success probability is at least $1-(2n+1)\delta$.  So to succeed with probability at least $1-\delta$ we must replace $\delta$ with $\delta/(2n+1)$.  Then the overall number of queries is
\begin{multline}
2n \cdot O\l(\frac{1}{\eps'^2} (2n)^{k-1} \log \frac{(2n+1)^{k}}{\delta}\r) + 1 \\= O\l(\frac{1}{\eps'^2} (2n)^{k} \log \frac{(2n+1)^{k}}{\delta}\r).
\end{multline}
\end{proof}
We remark that there is only $O(k \log n)$ overhead (for constant $\eps'$ and $\delta$) over the exact learning algorithm of \thmref{CkLearning}.

\section{Clifford Testing}
\label{sec:CliffordTesting}

Here we present an efficient algorithm to determine whether an unknown unitary operation is close to a Clifford or far from every Clifford.  Whereas the previous results allow us to find the Clifford operator close to the given black box unitary, in this section we are concerned with determining how far the given unitary is from any Clifford.  We do not measure this directly, but provide an algorithm of low query complexity that decides if the given unitary is close to a Clifford or far from all.  This type of algorithm is known in computer science as a property testing algorithm and has many applications, including the theory of probabilistically checkable proofs \cite{PCPTheorem}.  The result in this section could be extended to work for any level of the Gottesman-Chuang hierarchy although for simplicity we only present the version for Cliffords.

The key ingredient to our method will be a way of estimating the Pauli coefficients:
\begin{lemma}[Lemma 23 of \cite{QBF}]
\label{lem:MeasurePauliCoeff}
For any $p \in \{I, x, y, z\}^n$ and unitary $U$, $|\gamma(p)| = \frac{1}{2^n} \l| \tr U \sigma_p \r|$ can be estimated to within $\pm \eta$ with probability $1-\delta$ using $O\l(\frac{1}{\eta^2} \log \frac{1}{\delta}\r)$ queries.
\end{lemma}
This is a generalisation of \thmref{PauliLearning} and the method is similar.  Instead of there being only one possible outcome, now the probability of obtaining the outcome corresponding to $\sigma_p$ is estimated.  This probability is equal to $|\gamma(p)|^2$.
\begin{theorem}
\label{thm:CliffordTesting}
Given oracle access to $U$ and $U^\dagger$ with the promise that for $0 < \eps < 1$ either
\begin{enumerate}
\item[a)]{CLOSE: there exists $C \in \cC$ such that $D(U, C) \le \frac{\eps}{\sqrt{32}n}$ or}
\item[b)]{FAR: for all $C \in \cC$, $D(U,C) > \eps$ and there exists $C \in \cC$ such that $D(U, C) \le 1/3$}
\end{enumerate}
holds then there is a $O\l(\frac{n^3}{\eps^2} \log \frac{n}{\delta}\r)$ algorithm that determines which with probability at least $1-\delta$.
\end{theorem}
\begin{proof}
In both cases, we have that $D(U, C) < 1/3$ for some $C$, which ensures that $C$ is unique (using \lemref{UniqueCloseC}, since $\frac{1}{3} < \frac{1}{2 \sqrt{2}}$) and can be found using \thmref{ApproxLearning} with $O\l(n \log \frac{n}{\delta} \r)$ queries.  Then the algorithm is:
\begin{enumerate}
\item{For each $\sigma_g \in G$, measure the Pauli coefficient of $C \sigma_g C^\dagger$ in $U \sigma_g U^\dagger$ (i.e.~measure $\l| \tr U \sigma_g U^\dagger C \sigma_g C^\dagger \r|/2^n$) to precision $\frac{\eps^2}{16 n^2}$ using \lemref{MeasurePauliCoeff}.}
\item{If all the coefficients are found to have modulus at least $1 - \frac{3 \eps^2}{16 n^2}$ then output \emph{CLOSE} else output \emph{FAR}.}
\end{enumerate}
This works because, for the two possibilities \emph{CLOSE} and \emph{FAR}:
\begin{enumerate}
\item[a)]{Using \lemref{CloseHaveClosePauli}, $D(U,C) \le \frac{\eps}{\sqrt{32} n}$ implies that for all $\sigma_p \in \hat{\cP},$
\be D(U \sigma_p U^\dagger, C \sigma_p C^\dagger) \le \frac{2\eps}{\sqrt{32}n}. \ee
Since we will only apply $U \sigma_g U^\dagger$ for $\sigma_g \in G$ we restrict this to only the generators to find that for all $\sigma_g \in G,$
\be D(U \sigma_g U^\dagger, C \sigma_g C^\dagger) \le \frac{2\eps}{\sqrt{32}n} \ee
giving
\be \l|\frac{\tr U \sigma_g U^\dagger C \sigma_g C^\dagger}{2^n} \r|^2 \ge 1 - \frac{\eps^2}{8 n^2}
\ee
for every generator $\sigma_g$.  We need a bound on the non-squared coefficients, which follows directly:
\be \l|\frac{\tr U \sigma_g U^\dagger C \sigma_g C^\dagger}{2^n} \r| \ge 1 - \frac{\eps^2}{8 n^2}.
\ee
Therefore when measuring the coefficients to precision $\frac{\eps^2}{16 n^2}$, all results will give at least $1 - \frac{3\eps^2}{16 n^2}$.}
\item[b)]{
Using the contrapositive of \lemref{ClosePauliAreClose}, $D(U, C) > \eps$ implies that there exists $\sigma_p \in \hat{\cP}$ such that
\be D^+(U \sigma_p U^\dagger, C \sigma_p C^\dagger) > \eps. \ee
Using the contrapositive of \lemref{CloseGenAreClosePauli} this in turn implies there exists $\sigma_g \in G$ such that
\be D^+(U \sigma_g U^\dagger, C \sigma_g C^\dagger) > \frac{\eps}{2n}, \ee
which means that for at least one $\sigma_g \in G$, $U \sigma_g U^\dagger$ will have a small overlap with $C \sigma_g C^\dagger$ i.e.~there exists $\sigma_g \in G$ such that
\be
\l| \frac{\tr U \sigma_g U^\dagger C \sigma_g C^\dagger}{2^n} \r| < 1- \frac{\eps^2}{4n^2}.
\ee
The $C$ returned by the application of \thmref{ApproxLearning} is such that $\tr U \sigma_g U^\dagger C \sigma_g C^\dagger$ is positive, which justifies inserting the absolute value signs above when using $D^+$ rather than $D$.  This implies that at least one coefficient will be found to be less than $1- \frac{3\eps^2}{16 n^2}$ when measuring to precision $\frac{\eps^2}{16n^2}$.\qedhere
}
\end{enumerate}
\end{proof}

\section{Conclusions and Further Work}
\label{sec:Conclusions}

We have shown how to exactly identify an unknown Clifford operator in $O(n)$ queries, which we show is optimal.  This is then extended to cover elements of the $\cC_k$ hierarchy and for unitaries that are only known to be close to $\cC_k$ operations.  The key to the Clifford learning algorithm is to apply $C \sigma_p C^\dagger$ and then find the resulting Pauli operator.

A way of extending this idea could be to learn unitaries from larger sets.  Suppose $\cS$ is a set of unitaries with the property that for every $S \in \cS$, $S \sigma_p S^\dagger$ is a linear combination of a constant number of Paulis.  Then $S$ can be learnt in the same way as above, using the quantum Goldreich-Levin algorithm of \cite{QBF}, which can efficiently find which Paulis have large overlap with an input unitary.  However, we have not been able to find interesting sets $\cS$ other than the Clifford group with this property.

We also presented a Clifford testing algorithm, which determines whether a given black-box unitary is close to a Clifford or far from every Clifford.  This can be seen as a quantum generalisation of quadratic testing, just as Pauli testing can be seen as a quantum generalisation of linearity testing.  Property testing of this form is used to prove the PCP theorem \cite{PCPTheorem} so these quantum testing results could potentially be useful in proving a quantum PCP theorem.  It would also be interesting to strengthen the testing method in \thmref{CliffordTesting} to remove the $O(1/n)$ difference between the close and far conditions.

Finally, it would be interesting to see if it is possible to remove the requirement to have access to $U^\dagger$.  However, using both $U$ and $U^\dagger$ is the key to our method so we do not know if a method without $U^\dagger$ is possible with low query complexity.

\begin{acknowledgments}
I am grateful for funding from the U.K. Engineering and Physical Science Research Council through ``QIP IRC.''   I thank Aram Harrow for helpful discussions and suggestions and for comments on an earlier draft of this manuscript.  I also thank Ashley Montanaro for useful discussions and comments on the manuscript and Michael Bremner for helpful comments.
\end{acknowledgments}

\appendix
\section{Proof of \lemref{UniqueCloseC}}
\label{sec:ProofLemUniqueCloseC}

\begin{proof}[Proof of \lemref{UniqueCloseC}]
The proof is by induction.  The base case is for $k=1$ when we have the Pauli group.  Without loss of generality, assume $C$ is a Pauli operator with no phase.  Let $C = \sigma_p$.

Expand $U$ in the Pauli basis:
\be
U = \sum_q \gamma(q) \sigma_q.
\ee
Since $U$ is unitary, we have $\sum_q | \gamma(q) |^2 = 1$.  By \lemref{DistFromInnerProd},
\be
D(U, \sigma_p)^2 = 1 - \l| \frac{\tr \sigma_p U}{2^n} \r|^2
\ee
which implies
\be
|\gamma(p)|^2 \ge 1- \eps^2.
\ee

Now, suppose for contradiction that there exists $\sigma_{p_1} \ne \sigma_{p_2}$ with $D(U, \sigma_{p_1}) \le \eps$ and $D(U, \sigma_{p_2}) \le \eps$.  Then by the above, $|\gamma(p_1)|^2, |\gamma(p_2)|^2 \ge 1- \eps^2$.  But there is also the constraint $|\gamma(p_1)|^2 + |\gamma(p_2)|^2 \le 1$ which combined give
\be
\eps \ge \frac{1}{\sqrt{2}}
\ee
which is false by assumption.  This implies $\sigma_{p_1} = \sigma_{p_2}$, which proves the base case.

To prove the inductive step, again assume for contradiction that there exist $C_1, C_2 \in \cC_{k+1}$ with $C_1 \ne C_2$ and $D(U, C_1) \le \eps$ and $D(U, C_2) \le \eps$.  Then there exists $\sigma_g \in G$ with
\be
C_1 \sigma_g C_1^\dagger =: C_{1g} \ne C_{2g} := C_2 \sigma_g C_2^\dagger.
\ee
Here, $C_{1g}, C_{2g} \in \cC_k$.

Using \lemref{CloseHaveClosePauli}, $D(U \sigma_g U^\dagger, C_{1g}) \le 2\eps$ and $D(U \sigma_g U^\dagger, C_{2g}) \le 2\eps$.

Now there are two cases.  Firstly, suppose we can choose $\sigma_g$ such that $C_1 \sigma_g C_1^\dagger \ne \pm C_2 \sigma_g C_2^\dagger$.  Then $C_{1g}$ and $C_{2g}$ are not equivalent up to phase so, using the inductive hypothesis, we must have
\be
2 \eps \ge \frac{1}{2^{k-1/2}}
\ee
or
\be
\eps \ge \frac{1}{2^{(k+1)-1/2}}
\ee
which is again false by assumption.

For the other case, $C_1 \sigma_g C_1^\dagger = \pm C_2 \sigma_g C_2^\dagger$ for all $\sigma_g \in G$.  This implies that $C_2 = C_1 \sigma_q$ for some Pauli $\sigma_q \ne I$.  Then we have 
\ba
D(U, C_1) &\le \eps\nonumber\\
D(U, C_1 \sigma_q) &\le \eps
\ea
which by unitary invariance gives
\ba
D(C_1^\dagger U, I) &\le \eps\nonumber\\
D(C_1^\dagger U, \sigma_q) &\le \eps.
\ea
But we proved that this is impossible in this range of $\eps$ in the $k=1$ proof above.
\end{proof}

\section{Miscellaneous Lemmas}

Here we prove some miscellaneous lemmas used earlier in the paper.

The first lemma says that for two close operators $U_1$ and $U_2$, $U_1 \sigma_p U_1^\dagger$ is close to $U_2 \sigma_p U_2^\dagger$ for all Paulis $\sigma_p$:
\begin{lemma}
\label{lem:CloseHaveClosePauli}
If $D(U_1, U_2) \le \delta$ then for all $\sigma_p \in \hat{\cP}$,
\bes
D(U_1 \sigma_p U_1^\dagger, U_2 \sigma_p U_2^\dagger) \le 2 \delta.
\ees
\end{lemma}
\begin{proof}
Let $U_1 = V U_2$ and $U_{2p} = U_2 \sigma_p U_2^\dagger$.  Then we simply apply the triangle inequality for $D$ and unitary invariance:
\bas
D(U_1 \sigma_p U_1^\dagger, U_2 \sigma_p U_2^\dagger) &= D(V U_{2p} V^\dagger, U_{2p}) \\
&= D(V U_{2p}, U_{2p} V) \\
&\le D(V U_{2p}, U_{2p}) + D(U_{2p}, U_{2p} V) \\
&= D(V, I) + D(I, V) \\
&= 2D(U_1, U_2).\qedhere
\eas
\end{proof}
The next lemma is a converse to this:
\begin{lemma}
\label{lem:ClosePauliAreClose}
If for all $\sigma_p \in \hat{\cP}$
\be D^+(U_1 \sigma_p U_1^\dagger, U_2 \sigma_p U_2^\dagger) \le \delta \ee
then
\be
D(U_1, U_2) \le \delta.
\ee
\end{lemma}
\begin{proof}
If $D^+(U_1 \sigma_p U_1^\dagger, U_2 \sigma_p U_2^\dagger) \le \delta$ then $\frac{1}{2^n} \Re \tr U_1 \sigma_p U_1^\dagger U_2 \sigma_p U_2^\dagger \ge 1-\delta^2$.  Since this is true for all $\sigma_p$, we can take the average of this over the whole of $\hat{\cP}$ and use the fact that for any $d \times d$ matrix $A$ $\frac{1}{4^n}\sum_{\sigma_p \in \hat{\cP}} \sigma_p A \sigma_p = \frac{I}{2^n} \tr A$ (the Paulis are a \emph{1-design}) to find
\be
\frac{1}{2^n} \Re \tr U_1 \l( \frac{I}{2^n} \tr U_1^\dagger U_2 \r) U_2^\dagger \ge 1 - \delta^2
\ee
which simplified gives
\be
\l| \frac{\tr U_1 U_2^\dagger}{2^n} \r|^2 \ge 1 - \delta^2
\ee
giving the desired result.
\end{proof}
Now we show how to go from distances for just the generators $G$ to distances for the whole of $\hat{\cP}$:
\begin{lemma}
\label{lem:CloseGenAreClosePauli}
If for all $\sigma_g \in G$
\be
D^+(U_1 \sigma_g U_1^\dagger, U_2 \sigma_g U_2^\dagger) \le \delta
\ee
then for all $\sigma_p \in \hat{\cP}$
\be
D^+(U_1 \sigma_p U_1^\dagger, U_2 \sigma_p U_2^\dagger) \le 2n\delta
\ee
\end{lemma}
\begin{proof}
The proof is by induction on the number of generators required to make $\sigma_p$, using the triangle inequality for $D^+$.
\end{proof}


\begin{thebibliography}{17}
\expandafter\ifx\csname natexlab\endcsname\relax\def\natexlab#1{#1}\fi
\expandafter\ifx\csname bibnamefont\endcsname\relax
  \def\bibnamefont#1{#1}\fi
\expandafter\ifx\csname bibfnamefont\endcsname\relax
  \def\bibfnamefont#1{#1}\fi
\expandafter\ifx\csname citenamefont\endcsname\relax
  \def\citenamefont#1{#1}\fi
\expandafter\ifx\csname url\endcsname\relax
  \def\url#1{\texttt{#1}}\fi
\expandafter\ifx\csname urlprefix\endcsname\relax\def\urlprefix{URL }\fi
\providecommand{\bibinfo}[2]{#2}
\providecommand{\eprint}[2][]{\url{#2}}

\bibitem[{\citenamefont{Nielsen and Chuang}(2000)}]{NielsenChuang}
\bibinfo{author}{\bibfnamefont{M.~A.} \bibnamefont{Nielsen}} \bibnamefont{and}
  \bibinfo{author}{\bibfnamefont{I.~L.} \bibnamefont{Chuang}},
  \emph{\bibinfo{title}{Quantum Computation and Quantum Information}}
  (\bibinfo{publisher}{{Cambridge University Press}}, \bibinfo{year}{2000}).

\bibitem[{\citenamefont{{Chuang} and {Nielsen}}(1997)}]{ChuangNielsen97}
\bibinfo{author}{\bibfnamefont{I.~L.} \bibnamefont{{Chuang}}} \bibnamefont{and}
  \bibinfo{author}{\bibfnamefont{M.~A.} \bibnamefont{{Nielsen}}},
  \bibinfo{journal}{Journal of Modern Optics} \textbf{\bibinfo{volume}{44}},
  \bibinfo{pages}{2455} (\bibinfo{year}{1997}),
  \eprint{arXiv:quant-ph/9610001}.

\bibitem[{\citenamefont{Poyatos et~al.}(1997)\citenamefont{Poyatos, Cirac, and
  Zoller}}]{PoyatosCiracZoller97}
\bibinfo{author}{\bibfnamefont{J.~F.} \bibnamefont{Poyatos}},
  \bibinfo{author}{\bibfnamefont{J.~I.} \bibnamefont{Cirac}}, \bibnamefont{and}
  \bibinfo{author}{\bibfnamefont{P.}~\bibnamefont{Zoller}},
  \bibinfo{journal}{Phys. Rev. Lett.} \textbf{\bibinfo{volume}{78}},
  \bibinfo{pages}{390} (\bibinfo{year}{1997}), \eprint{arXiv:quant-ph/9611013}.

\bibitem[{\citenamefont{D'Ariano and Lo~Presti}(2001)}]{DArianoPresti01}
\bibinfo{author}{\bibfnamefont{G.~M.} \bibnamefont{D'Ariano}} \bibnamefont{and}
  \bibinfo{author}{\bibfnamefont{P.}~\bibnamefont{Lo~Presti}},
  \bibinfo{journal}{Phys. Rev. Lett.} \textbf{\bibinfo{volume}{86}},
  \bibinfo{pages}{4195} (\bibinfo{year}{2001}),
  \eprint{arXiv:quant-ph/0012071}.

\bibitem[{\citenamefont{Leung}(2003)}]{Leung03}
\bibinfo{author}{\bibfnamefont{D.~W.} \bibnamefont{Leung}},
  \bibinfo{journal}{Journal of Mathematical Physics}
  \textbf{\bibinfo{volume}{44}}, \bibinfo{pages}{528} (\bibinfo{year}{2003}),
  \eprint{arXiv:quant-ph/0201119}.

\bibitem[{\citenamefont{H\"{a}ffner et~al.}(2005)\citenamefont{H\"{a}ffner,
  H\"{a}nsel, Roos, Benhelm, Chek-Al-Kar, Chwalla, K\"{o}rber, Rapol, Riebe,
  Schmidt et~al.}}]{HaffnerTomography}
\bibinfo{author}{\bibfnamefont{H.}~\bibnamefont{H\"{a}ffner}},
  \bibinfo{author}{\bibfnamefont{W.}~\bibnamefont{H\"{a}nsel}},
  \bibinfo{author}{\bibfnamefont{C.~F.} \bibnamefont{Roos}},
  \bibinfo{author}{\bibfnamefont{J.}~\bibnamefont{Benhelm}},
  \bibinfo{author}{\bibfnamefont{D.}~\bibnamefont{Chek-Al-Kar}},
  \bibinfo{author}{\bibfnamefont{M.}~\bibnamefont{Chwalla}},
  \bibinfo{author}{\bibfnamefont{T.}~\bibnamefont{K\"{o}rber}},
  \bibinfo{author}{\bibfnamefont{U.~D.} \bibnamefont{Rapol}},
  \bibinfo{author}{\bibfnamefont{M.}~\bibnamefont{Riebe}},
  \bibinfo{author}{\bibfnamefont{P.~O.} \bibnamefont{Schmidt}},
  \bibnamefont{et~al.}, \bibinfo{journal}{Nature}
  \textbf{\bibinfo{volume}{438}}, \bibinfo{pages}{643} (\bibinfo{year}{2005}).

\bibitem[{\citenamefont{Bennett and Wiesner}(1992)}]{SuperdenseCoding}
\bibinfo{author}{\bibfnamefont{C.~H.} \bibnamefont{Bennett}} \bibnamefont{and}
  \bibinfo{author}{\bibfnamefont{S.~J.} \bibnamefont{Wiesner}},
  \bibinfo{journal}{Phys. Rev. Lett.} \textbf{\bibinfo{volume}{69}},
  \bibinfo{pages}{2881} (\bibinfo{year}{1992}).

\bibitem[{\citenamefont{Calderbank et~al.}(1997)\citenamefont{Calderbank,
  Rains, Shor, and Sloane}}]{QECGeometry}
\bibinfo{author}{\bibfnamefont{A.~R.} \bibnamefont{Calderbank}},
  \bibinfo{author}{\bibfnamefont{E.~M.} \bibnamefont{Rains}},
  \bibinfo{author}{\bibfnamefont{P.~W.} \bibnamefont{Shor}}, \bibnamefont{and}
  \bibinfo{author}{\bibfnamefont{N.~J.~A.} \bibnamefont{Sloane}},
  \bibinfo{journal}{Phys. Rev. Lett.} \textbf{\bibinfo{volume}{78}},
  \bibinfo{pages}{405} (\bibinfo{year}{1997}).

\bibitem[{\citenamefont{Gottesman}(1998)}]{GottesmanFaultTolerance}
\bibinfo{author}{\bibfnamefont{D.}~\bibnamefont{Gottesman}},
  \bibinfo{journal}{Phys. Rev. A} \textbf{\bibinfo{volume}{57}},
  \bibinfo{pages}{127} (\bibinfo{year}{1998}), \eprint{arXiv:quant-ph/9702029}.

\bibitem[{\citenamefont{Shor}(1996)}]{ShorFaultTolerance}
\bibinfo{author}{\bibfnamefont{P.~W.} \bibnamefont{Shor}},
  \bibinfo{journal}{Proceedings of the 37th Annual Symposium on Foundations of
  Computer Science} pp. \bibinfo{pages}{56--65} (\bibinfo{year}{1996}),
  \eprint{arXiv:quant-ph/9605011}.

\bibitem[{\citenamefont{{Gottesman} and {Chuang}}(1999)}]{GottesmanChuang}
\bibinfo{author}{\bibfnamefont{D.}~\bibnamefont{{Gottesman}}} \bibnamefont{and}
  \bibinfo{author}{\bibfnamefont{I.~L.} \bibnamefont{{Chuang}}},
  \bibinfo{journal}{Nature} \textbf{\bibinfo{volume}{402}},
  \bibinfo{pages}{390} (\bibinfo{year}{1999}), \eprint{arXiv:quant-ph/9908010}.

\bibitem[{\citenamefont{{Montanaro} and {Osborne}}(2008)}]{QBF}
\bibinfo{author}{\bibfnamefont{A.}~\bibnamefont{{Montanaro}}} \bibnamefont{and}
  \bibinfo{author}{\bibfnamefont{T.~J.} \bibnamefont{{Osborne}}},
  \emph{\bibinfo{title}{{Quantum boolean functions}}} (\bibinfo{year}{2008}),
  \bibinfo{note}{arXiv:0810.2435}.

\bibitem[{\citenamefont{{Aaronson}}(2007)}]{AaronsonLearnability}
\bibinfo{author}{\bibfnamefont{S.}~\bibnamefont{{Aaronson}}},
  \bibinfo{journal}{Royal Society of London Proceedings Series A}
  \textbf{\bibinfo{volume}{463}}, \bibinfo{pages}{3089} (\bibinfo{year}{2007}),
  \eprint{arXiv:quant-ph/0608142}.

\bibitem[{\citenamefont{{Aaronson} and
  {Gottesman}}(2009)}]{AaronsonGottesmanStabilisers}
\bibinfo{author}{\bibfnamefont{S.}~\bibnamefont{{Aaronson}}} \bibnamefont{and}
  \bibinfo{author}{\bibfnamefont{D.}~\bibnamefont{{Gottesman}}}
  (\bibinfo{year}{2009}), \bibinfo{note}{{Unpublished}}.

\bibitem[{\citenamefont{Calderbank et~al.}(1998)\citenamefont{Calderbank,
  Rains, Shor, and Sloane}}]{Calderbank98quantumerror}
\bibinfo{author}{\bibfnamefont{A.~R.} \bibnamefont{Calderbank}},
  \bibinfo{author}{\bibfnamefont{E.~M.} \bibnamefont{Rains}},
  \bibinfo{author}{\bibfnamefont{P.~W.} \bibnamefont{Shor}}, \bibnamefont{and}
  \bibinfo{author}{\bibfnamefont{N.~J.~A.} \bibnamefont{Sloane}},
  \bibinfo{journal}{IEEE Trans. Inform. Theory} \textbf{\bibinfo{volume}{44}},
  \bibinfo{pages}{1369} (\bibinfo{year}{1998}).

\bibitem[{\citenamefont{{Harrow} and {Winter}}(2006)}]{HowManyCopiesStateDesc}
\bibinfo{author}{\bibfnamefont{A.~W.} \bibnamefont{{Harrow}}} \bibnamefont{and}
  \bibinfo{author}{\bibfnamefont{A.}~\bibnamefont{{Winter}}}
  (\bibinfo{year}{2006}), \bibinfo{note}{arXiv:quant-ph/0606131}.

\bibitem[{\citenamefont{{Arora} et~al.}(1998)\citenamefont{{Arora}, {Lund},
  {Motwani}, {Sudan}, and {Szegedy}}}]{PCPTheorem}
\bibinfo{author}{\bibfnamefont{S.}~\bibnamefont{{Arora}}},
  \bibinfo{author}{\bibfnamefont{C.}~\bibnamefont{{Lund}}},
  \bibinfo{author}{\bibfnamefont{R.}~\bibnamefont{{Motwani}}},
  \bibinfo{author}{\bibfnamefont{M.}~\bibnamefont{{Sudan}}}, \bibnamefont{and}
  \bibinfo{author}{\bibfnamefont{M.}~\bibnamefont{{Szegedy}}},
  \bibinfo{journal}{J. ACM} \textbf{\bibinfo{volume}{45}}, \bibinfo{pages}{501}
  (\bibinfo{year}{1998}).

\end{thebibliography}
\end{document}